\pdfoutput=1
\documentclass[sigconf,natbib=true,nonacm]{acmart}
\AtBeginDocument{%
  }
\usepackage[utf8]{inputenc} 
\usepackage[T1]{fontenc}    
\usepackage{url}            
\usepackage{nicefrac}       
\usepackage{microtype}      
\usepackage{xcolor}         
\usepackage{bm}
\usepackage{bbm}
\usepackage{todonotes}
\usepackage{algorithm}
\usepackage{algpseudocode}
\usepackage{multicol}
\usepackage{caption}
\usepackage{multirow}
\usepackage{makecell}
\usepackage{booktabs}           
\usepackage{amsfonts}           
\usepackage{graphicx}           
\usepackage{duckuments}         
\usepackage{bold-extra}
\usepackage{amsmath}
\usepackage{amsmath, mathtools, bm}
\usepackage[dvipsnames]{xcolor,colortbl}
\definecolor{Gray}{gray}{0.9}

\begin{document}

\title{CURE:Circuit-Aware Unlearning for LLM-based Recommendation}

\author{Ziheng Chen}
\email{albertchen1993pokemon@gmail.com}
\affiliation{%
  \institution{Walmart Global Tech}
  \city{Sunnyvale, CA}
  \country{USA}
}

\author{Jiali Cheng}
\email{jiali\_cheng@uml.edu}
\affiliation{%
  \institution{University of Massachusetts Lowell}
  \city{Lowell, MA}
  \country{USA}
}

\author{Zezhong Fan}
\email{zfan2274@gmail.com}
\affiliation{%
  \institution{Walmart Global Tech}
  \city{Sunnyvale, CA}
  \country{USA}
}

\author{Hadi Amiri}
\email{hadi_amiri@uml.edu}
\affiliation{%
  \institution{University of Massachusetts Lowell}
  \city{Lowell, MA}
  \country{USA}
}

\author{Yunzhi Yao}
\email{yyztodd@zju.edu.cn}
\affiliation{%
  \institution{Zhejiang University}
  \city{Hangzhou}
  \country{China}
}

\author{Xiangguo Sun}
\email{xiangguosun@cuhk.edu.hk}
\affiliation{%
  \institution{The Chinese University of Hong Kong}
  \city{Hong Kong}
  \country{China}
}

\author{Yang Zhang}
\email{zhangy@nus.edu.sg}
\affiliation{%
  \institution{National University of Singapore}
  \country{Signapore}
}

\newcommand{\X}{\mathcal{X}}
\newcommand{\Y}{\mathcal{Y}}
\newcommand{\R}{\mathbb{R}}
\newcommand{\paramspace}{\boldsymbol{\Theta}}
\newcommand{\params}{\boldsymbol{\theta}}
\newcommand{\model}{f_{\params}}
\newcommand{\modeledit}{f_{\params'}}
\newcommand{\inst}{\boldsymbol{x}}
\newcommand{\cfinst}{\boldsymbol{\widetilde{x}}}
\newcommand{\pred}{\hat{y}}
\newcommand{\instedit}{\inst_c}
\newcommand{\instq}{\inst_q}
\newcommand{\labeledit}{y_c}
\newcommand{\labelq}{y_q}
\newcommand{\prededit}{\hat{y}_c}
\newcommand{\editorparamspace}{\boldsymbol{\Psi}}
\newcommand{\editorparams}{\boldsymbol{\psi}}
\newcommand{\dataset}{\mathcal{D}}
\newcommand{\testset}{\dataset_{q}}
\newcommand{\datasetedit}{\dataset_c}
\newcommand{\loss}{\ell}
\newcommand{\Loss}{\mathcal{L}}
\newcommand{\method}{\textsc{FROG}\xspace}
\newcommand{\softmax}{\text{softmax}}

\newcommand{\Prob}{\mathbb{P}}
\newcommand{\Z}{\mathbb{Z}}
\newcommand{\E}{\mathbb{E}}
\newcommand{\G}{\mathcal{G}}
\newcommand{\advG}{\widetilde{\G}} 
\newcommand{\Gobs}{\G^{\text{obs}}}
\newcommand{\U}{\mathcal{U}}
\newcommand{\I}{\mathcal{I}}
\newcommand{\V}{\mathcal{V}}
\newcommand{\Vnew}{\V^{\text{new}}}
\newcommand{\Vadv}{\V^{\text{adv}}}
\newcommand{\edges}{\mathcal{E}}
\newcommand{\edgesnew}{\edges^{\text{new}}}
\newcommand{\edgesobs}{\edges^{\text{obs}}}
\newcommand{\edgesadv}{\edges^{\text{adv}}}
\newcommand{\bgraph}{\G=(\U, \I, \edges)}
\newcommand{\neigh}{\mathcal{N}}
\newcommand{\adjM}{A}
\newcommand{\advadjM}{\widetilde{A}}
\newcommand{\adjMij}{{A}_{i,j}}
\newcommand{\train}{\dataset_{\text{train}}}
\newcommand{\test}{\dataset_{\text{test}}}
\newcommand{\features}{\mathcal{X}}
\newcommand{\labels}{\mathcal{Y}}
\newcommand{\hypspace}{\mathcal{H}}
\newcommand{\w}{\bm{\omega}}
\newcommand{\h}{\bm{h}}
\newcommand{\advh}{\widetilde{\bm{h}}}
\newcommand{\hyp}{h_{\params}}
\newcommand{\gnn}{g(\adjM, \X; \W)}
\newcommand{\Cons}{\mathcal{L}_{\text{cstr}}}
\newcommand{\ladv}{\ell_{\text{adv}}}
\newcommand{\ldist}{\ell_{\text{dist}}}
\newcommand{\lnew}{\ell_{\text{new}}}
\newcommand{\LF}{\Loss_{fa}}
\newcommand{\LC}{\Loss_{cf}}
\newcommand{\ind}{\mathbbm{1}}
\newcommand{\inputs}{\mathcal{V}}
\newcommand{\outputs}{\mathcal{Y}}
\newcommand{\insta}{X_i}
\newcommand{\weights}{\boldsymbol{\omega}}
\newcommand{\forget}{\mathcal{D}_f}
\newcommand{\remain}{\mathcal{D}_r}
\newcommand{\all}{\mathcal{D}}
\newcommand{\LLMori}{\mathcal{M}_{\theta}}
\newcommand{\un}{\mathcal{M}_{un}}
\newcommand{\optimal}{\mathcal{\mathcal{M}}_{\theta^{*}}}
\newcommand{\indep}{\perp \!\!\! \perp}
\newcommand{\cnf}{\bold{Z_{cnf}}}
\newcommand{\graph}{\mathcal{G}}
\newcommand{\node}{v}
\newcommand{\edge}{\mathcal{E}}
\newcommand{\nshare}{\boldsymbol{\Tilde{\theta}_{sh}}}
\newcommand{\nforget}{\boldsymbol{\Tilde{\theta}_{f}}}
\newcommand{\nretain}{\boldsymbol{\Tilde{\theta}_{r}}}
\newcommand{\selshare}{\boldsymbol{\Hat{\theta}_{sh}}}
\newcommand{\selforget}{\boldsymbol{\Hat{\theta}_{f}}}
\newcommand{\selretain}{\boldsymbol{\Hat{\theta}_{r}}}
\newcommand{\orishare}{\boldsymbol{\theta}_{sh}}
\newcommand{\oriforget}{\boldsymbol{\theta}_{f}}
\newcommand{\oriretain}{\boldsymbol{\theta}_{r}}
\newcommand{\gradfshare}{\mathbf{g}_{sh}^{F}}
\newcommand{\gradfforget}{\mathbf{g}_{f}^{F}}
\newcommand{\gradfretain}{\mathbf{g}_{r}^{F}} 
\newcommand{\gradrshare}{\mathbf{g}_{sh}^{R}}
\newcommand{\gradrforget}{\mathbf{g}_{f}^{R}} 
\newcommand{\gradrretain}{\mathbf{g}_{r}^{R}}
\newcommand{\wforget}{\omega_f}
\newcommand{\wretain}{\omega_r}
\newcommand{\inpori}{\mathbf{x}_{u}}
\newcommand{\inpcf}{\mathbf{x}_{u}^{*}}

\newcommand{\Uset}{\mathcal{U}}
\newcommand{\Iset}{\mathcal{I}}
\newcommand{\LLM}{\mathcal{M}}


\newcommand{\xg}[1]{\todo[inline,color=red!60]{\textbf{xiangguo:} #1}}
\newcommand{\xgr}[1]{\textcolor{red}{\textbf{xiangguo:} #1}}

\newcommand{\zezhong}[1]{\todo[inline,color=pink!60]{\textbf{zezhong:} #1}}
\newcommand{\yang}[1]{\todo[inline,color=green!60]{\textbf{Yang:} #1}}
\newcommand{\jiali}[1]{\todo[inline,color=orange!60]{\textbf{Jiali:} #1}}
\newcommand{\ziheng}[1]{\todo[inline,color=blue!60]{\textbf{Ziheng:} #1}}
\newcommand{\hadi}[1]{\todo[inline,color=yellow!60]{\textbf{Hadi:} #1}}
\newcommand{\sun}[1]{\todo[inline,color=cyan!60]{\textbf{sun:} #1}}
\newcommand{\yunzhi}[1]{\todo[inline,color=magenta!60]{{\bf Yunzhi:} #1}}

\begin{abstract}
Recent advances in large language models (LLMs) have opened new opportunities for recommender systems by enabling rich semantic understanding and reasoning about user interests and item attributes. However, as privacy regulations tighten, incorporating user data into LLM-based recommendation (LLMRec) introduces significant privacy risks, making unlearning algorithms increasingly crucial for practical deployment. Despite growing interest in LLMRec unlearning, most existing approaches formulate unlearning as a weighted combination of forgetting and retaining objectives while updating model parameters in a uniform manner. Such formulations inevitably induce gradient conflicts between the two objectives, leading to unstable optimization and resulting in either ineffective unlearning or severe degradation of model utility. Moreover, the unlearning procedure remains largely black-box, undermining its transparency and trustworthiness.
\par\indent
To tackle these challenges, we propose CURE, a circuit-aware unlearning framework that disentangles model components into functionally distinct subsets and selectively updates them. Here, a circuit refers to a computational subgraph that is causally responsible for task-specific behaviors.
Specifically, we extract the core circuits underlying item recommendation and analyze how individual modules within these circuits contribute to the forget and retain objectives.
Based on this analysis, these modules are categorized into \textit{forget-specific}, \textit{retain-specific}, and \textit{task-shared} groups, each subject to function-specific update rules to mitigate gradient conflicts during unlearning. Experiments on real-world datasets show that our approach achieves more effective unlearning than existing baselines.
\end{abstract}
\maketitle

\section{Introduction}
\label{sec:intro} 
Large language models (LLMs) have recently shown remarkable capability in understanding and generating human-like text, which has spurred growing interest in leveraging them as recommender systems (LLMRec).
By exploiting powerful reasoning abilities and rich open-world knowledge, LLMRecs can better model user preferences and item semantics through instruction tuning on historical interactions. However, directly incorporating user behavior data, such as purchase records, raises critical ethical and privacy concerns, including information leakage and the risk of malicious data injection. To address these challenges, recommendation unlearning~\cite{wang2025towards,hu2025exact} has emerged as a promising paradigm that aims to remove the influence of sensitive data from pre-trained LLMRecs, while preserving their overall utility.

Existing methods for LLMRec unlearning can be categorized into two types: approximate unlearning and exact unlearning. Exact unlearning relies on retraining affected sub-models~\citep{bourtoule2021machine}, while approximate unlearning typically adopts a teacher–student framework to balance data removal and model utility~\citep{frog,fan2025towards,cheng2025tool}. However, most approximate unlearning methods 
formulate unlearning as a weighted sum of the forget and retain losses, using a static factor to balance the two objectives~\citep{fan2023salun,cheng2024mu,fan2025towards}. Despite its simplicity, this formulation overlooks the fact that optimizing one objective can substantially impede the other. From the optimization point of view, a key cause of this issue lies in \textit{conflicting gradients}~\cite{yu2020gradient,yi2025gradient}, where the gradients of the forget and retain losses at the neuron level point in opposing directions. As a result, updates intended to improve one objective may inadvertently harm the other, leading to either insufficient forgetting or severe degradation of model utility

Additionally, existing LLMRec unlearning methods update model parameters in a largely black-box manner, without explicit knowledge of which internal modules (e.g., attention heads and MLPs) encode the information to be forgotten~\citep{chen2022recommendation,chen2024cure4rec}. As a result, it is unclear whether modules containing critical information are effectively updated. This lack of transparency hinders the interpretability of the unlearning process and undermines its trustworthiness.

Recent advances in mechanistic interpretability shed light on the internal mechanisms of LLMs by identifying sparse computational subnetworks (“circuits”) responsible for specific model behaviors, offering a reliable way to understand the above challenges at their root. A key finding is that knowledge in LLM is dynamically activated through specific computational circuits, each specializing in different functional roles and jointly contributing to the final decision~\cite{conmy2023towards,syed2024attribution,cheng2026toward}. Consequently, gradient conflicts arise when the circuits responsible for the forget and retain sets become entangled, particularly when the shared modules are driven toward conflicting optimization directions under the two objectives. This insight suggests a transparent solution: disentangling the conflicting circuits and optimizing them separately for unlearning. 

Inspired by this insight, we propose CURE, a circuit-aware framework for LLMRec unlearning. Instead of globally optimizing competing objectives, CURE decouples unlearning into two stages: \textit{crucial circuit extraction} and \textit{task-specific parameter updating}. In the first stage, CURE employs a gradient-based analysis to localize internal computational pathways that are most responsible for the forget and retain sets. To precisely identify these circuits under long input prompts encoding user interaction histories, we leverage the user–item graph to construct slight input perturbations and detect influential modules through contrastive activation analysis. In the second stage, modules along the identified circuits are categorized according to their functional roles and selectively updated, enabling semantically aware control over parameters while effectively mitigating gradient conflicts during unlearning. We also demonstrate the effectiveness of CURE through a theoretical analysis. Overall, our key contributions are as follows:
\begin{itemize}
\item \textbf{Motivation:} We introduce a circuit-aware perspective for understanding gradient conflicts in LLMRec unlearning, which leads to a transparent unlearning approach via disentangling conflicting circuits. 


\item \textbf{Method:} We propose a novel circuit-aware framework for LLMRec unlearning that first identifies influential circuits for the forget and retain sets, and then selectively updates the associated modules according to their functional roles. We further provide a theoretical analysis to guarantee its validity. Furthermore, the proposed framework is model-agnostic and can be readily applied to a wide range of LLM backbones for recommendation unlearning. 
\item \textbf{Performance:} CURE achieves $18\%$ and $6\%$ improvements over the baseline in unlearning efficiency and model utility, respectively. Moreover, it is $\mathbf{3.5}\times$ faster than the baseline.
\end{itemize}
\section{Preliminary}
\label{sec:Preliminary}
\subsection{LLM as Recommender}
\label{LLMREC}
We consider the standard collaborative filtering recommendation task.
Let $\mathcal{U}=\{u_1,\ldots,u_m\}$ be a set of $m$ users, and $\mathcal{I}=\{i_1,\ldots,i_n\}$ be a set of $n$ items.
The observed user–item interactions are represented by a binary interaction matrix $A\in{0,1}^{m\times n}$, where $A_{u,i}=1$ indicates that user $u$ has interacted with item $i$, and $A_{u,i}=0$ otherwise. We further model the interaction data as user-item graph $
\mathcal{G}=(\mathcal{V},\mathcal{E})$, where $\mathcal{V}=\mathcal{U}\cup\mathcal{I}$ and $\mathcal{E}=\{(u,i)\mid A_{u,i}=1\}$ denotes the set of edges corresponding to observed interactions.

Based on the user-item graph $\mathcal{G}$, LLMRec, denoted as $\LLMori$, reformulates collaborative filtering as a prompt-based prediction task. Given a user $u\in\Uset$ with historical interactions $\mathcal{H}_u=\{i_1,i_2,\cdots,i_{|\mathcal{H}_u|}\}$ and a target item $i_t\in\Iset$, we encode them into textual instructions $\inpori$ using predefined hard prompt templates. Conditioned on $\inpori$, $\LLMori$ predicts whether $u$ will interact with $i_t$, which is formulated as a binary classification task with an answer $y_{u} \in \{\mbox{"YES"},\mbox{"No"}\}$(As shown in Figure~\ref{fig:CURE} and ~\ref{fig:casestudy}).
In order to tailor LLM to recommendation scenarios, a conditional language modeling objective is employed by minimizing the negative log-likelihood of generating $y_{u}$ conditioned on input $\inpori$. Formally:
\begin{equation}
\label{eq:LLMRec}
\begin{aligned}
\min \mathcal{L}_{pred}= -\sum\limits_{(\inpori,y_u)\in\all}\mbox{log}(\LLMori(y_{u,t}|\inpori,y_{u,<t}))
\end{aligned}
\end{equation}
Where $y_{u,t}$ is the \textit{t}-th token of $y_u$, and $y_{u,<t}$ is the token before $y_{u,t}$. $\LLMori(y_{u,t}|\inpori,y_{u,<t})$ signifies the predictive probability of $y_{u,t}$.


\subsection{LLMRec Unlearning}

LLMRec unlearning involves removing certain user-item interactions from a trained model $\LLMori$ without full retraining. Given a dataset $\all$
and a subset $\forget$ to be removed, we denote the retained dataset as $\remain$, where $\remain = \all \setminus \forget$,
with the conditions $\forget \cap \remain = \emptyset$. Requests for LLMRec unlearning can be broadly categorized into two types:
(i) \textit{\textbf{user/item-wise deletion}}, , which removes all interactions associated with a given user or item, and
(ii) \textit{\textbf{interaction deletion}}, which removes specific interactions.

The objective is to obtain an unlearned model $\un$ that eliminates the influence of $\forget$ while preserving performance on $\remain$. As retraining on $\remain$ to obtain the optimal model $\optimal$ is often time-consuming, our goal is to approximate $\optimal$ by updating the original model $\LLMori$ through the unlearning process as follows:
\begin{equation*}
    \LLMori \xrightarrow{\forget} \un \approx \optimal.
\end{equation*}

Most existing methods formulate the LLMRec unlearning loss $\mathcal{L}$ as a weighted sum of a forget loss $\mathcal{L}_{\mathrm{F}}$ and a retain loss $\mathcal{L}_{\mathrm{R}}$. 
More formally, the task of unlearning can be modeled in the following manner:
\begin{equation}
\label{eq:Weightsum}
\begin{aligned}
\un &= \mbox{arg} \min\limits_{\mathcal{M}}\sum\limits_{(\inpori,y_u)\in \forget}\wforget\mathcal{L}_{\mathrm{F}}(y_u|\inpori;\mathcal{M}) + \sum\limits_{(\inpori,y_u)\in \remain}\wretain\mathcal{L}_{\mathrm{R}}(y_u|\inpori;\mathcal{M}) 
\end{aligned}
\end{equation}
 In addition, $\wretain+\wforget=1$ serves as a scaling factor to balance forget and retain. The specific forms of $\mathcal{L}_{\mathrm{F}}$ and $\mathcal{L}_{\mathrm{R}}$ are introduced in section~\ref{sec:method}. Note that \textit{\textbf{user/item-wise deletion}} can be interpreted as removing all interactions incident involving the corresponding user or item.

\section{Motivation}
\label{sec:motivation}
Despite simplicity, Eq.~\ref{eq:Weightsum} fails to account for scenarios in which the gradients of the forget and retain losses move in opposing directions during gradient descent optimization~\citep{reisizadeh2025blur}, leading to either insufficient forgetting on $\forget$ or performance degradation on $\remain$. 

To illustrate this issue, we analyze the gradient structure of the unlearning objective. Specifically, let $g=\nabla \mathcal{L}(\theta)$ denote the gradient of the weighted unlearning loss with respect to the model parameters $\theta$, while $g_f=\nabla \mathcal{L}_{F}(\theta)$ and $g_r=\nabla \mathcal{L}_{R}(\theta)$ denote the gradient of forget loss and retain losses respectively. A small change of $\theta$ in the direction of negative $g$ is $\theta\leftarrow\theta-\alpha g$ with a sufficiently small step size $\alpha$. The effect of this change on the performance of two objective can be measured by:
\begin{equation}
\label{eq:Decomposing}
\begin{aligned}
\Delta\mathcal{L}
&=\wforget\Delta\mathcal{L}_F+\wretain\Delta\mathcal{L}_R\\
&=\wforget\mathcal{L}_{F}(\mathbf{\theta}-\alpha \mathbf{g})-\wforget\mathcal{L}_{F}(\mathbf{\theta})
  +\wretain\mathcal{L}_{R}(\mathbf{\theta}-\alpha \mathbf{g})-\wretain\mathcal{L}_{R}(\mathbf{\theta})\\
&=-\alpha\wforget\, (\mathbf{g}\cdot \mathbf{g}_f)
  -\alpha\wretain\, (\mathbf{g}\cdot \mathbf{g}_r)
  + o(\alpha),
\end{aligned}
\end{equation}
where the last equality is obtained by Taylor approximation. Notably, the update procedure can impede effective forgetting when $\mathbf{g}\cdot \mathbf{g}_f<0$, since it increases the forget loss. Similarly, it can degrade model utility when $\mathbf{g}\cdot \mathbf{g}_r<0$.

Based on the above analysis, we adopt a normalized alignment metric~\cite{reisizadeh2025blur} and examine the evolution of $A_r(\mathbf{\theta})$ and $A_f(\mathbf{\theta})$ throughout the unlearning procedure in Figure~\ref{fig:alignment}, where
\begin{equation}
\label{eq:normalized alignment}
\begin{aligned}
A_f(\mathbf{\theta})=\frac{\mathbf{g}\cdot \mathbf{g}_f}{||\mathbf{g}_f||^{2}}, \quad A_r(\mathbf{\theta})=\frac{\mathbf{g}\cdot \mathbf{g}_r}{||\mathbf{g}_r||^{2}}
\end{aligned}
\end{equation}

\begin{figure}
\centering
    \includegraphics[width=1.0\linewidth]{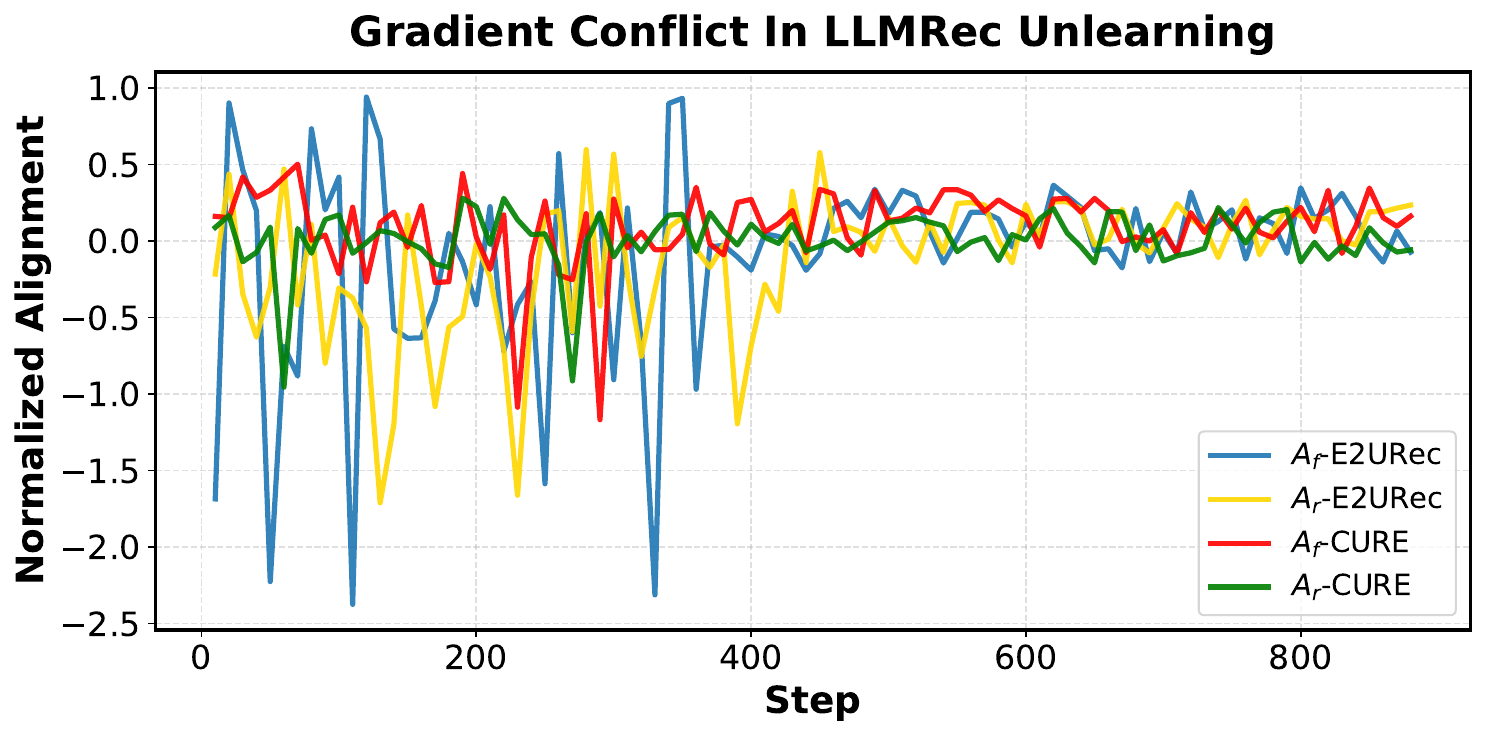}
    \caption{Normalized Alignment Values of forget and retain losses on MovieLens-1M using \textbf{Llama-2 (7B)}}

    \label{fig:alignment}
\end{figure}
We observe that the descent direction switches frequently during the early training steps of E2URec, a representative LLMRec unlearning baseline, reflecting entangled optimization signals and the presence of gradient conflicts.
Although \cite{reisizadeh2025blur} proposes a hierarchical unlearning framework that prioritizes forgetting over retention, such a trade-off is suboptimal for LLMRec unlearning, where both objectives are essential. In contrast, CURE exhibits a more stable optimization behavior, with gradient conflicts largely mitigated.

\section{Method}
\label{sec:method}
We introduce CURE, a two-stage circuit-aware framework for LLMRec unlearning. The first stage focuses on identifying key computational circuits in the original model . Leveraging the structural information encoded in 
$\LLMori$, we propose two alternatives methods: \textit{Activation Intervention}, which efficiently yields a coarse circuit estimate, and \textit{Activation Patching}, a contrastive approach for more precise localization. In the second stage, we selectively update circuit components according to their functional roles in forgetting and utility preservation, thereby mitigating gradient conflicts.

\begin{figure*}[t]
    \centering
    \includegraphics[width=1.0\linewidth,height=0.4\textheight]{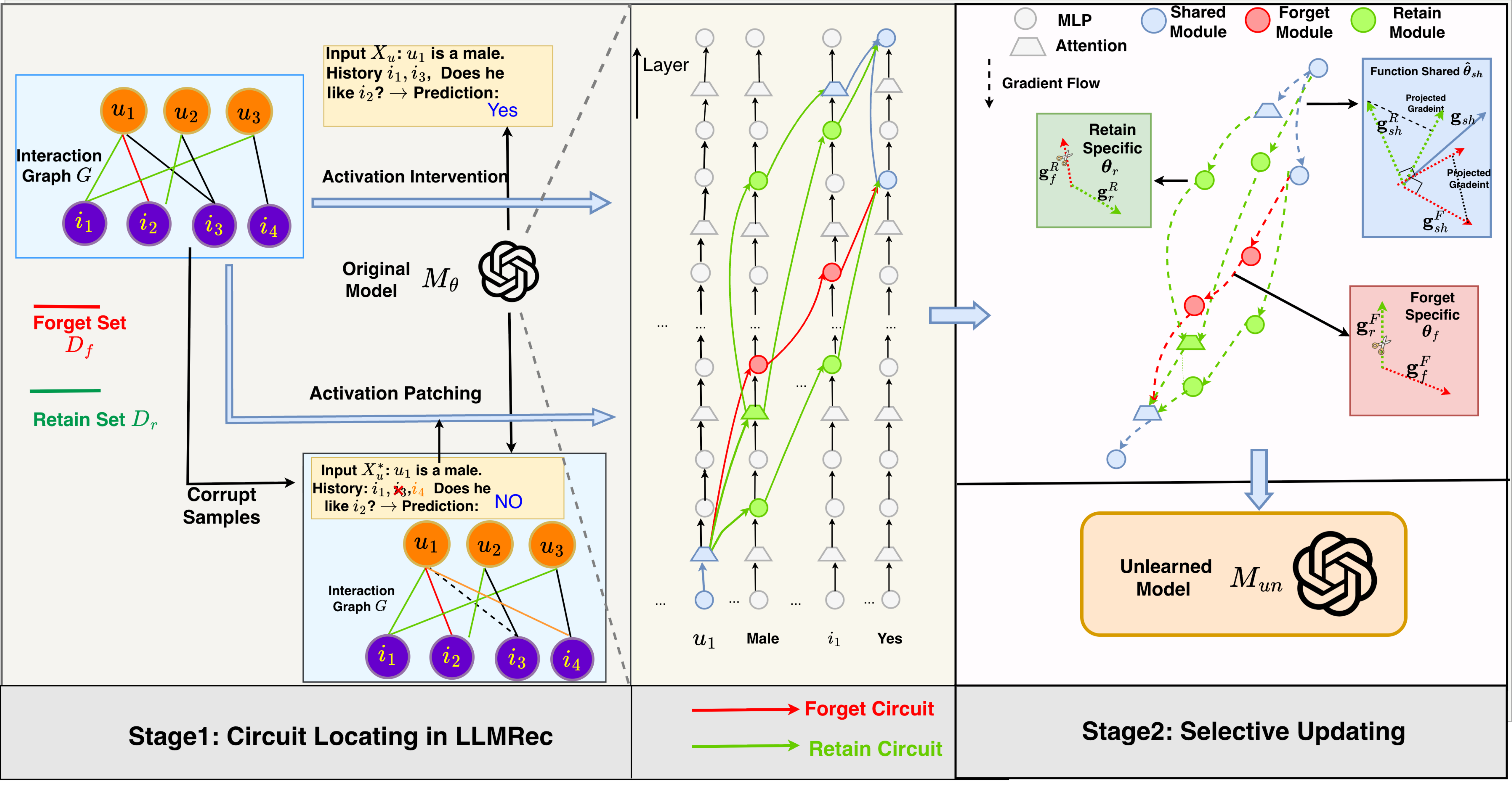}
    \caption{Schematics of CURE: (i) Locating Circuits in LLMREC; (ii)Selective Circuits Updating}
    \label{fig:CURE}
\end{figure*}


\label{sec}

\subsection{Locating Circuits in LLMRec}
\label{sec:Circuits}
To find circuits for LLMRec, we must represent the internals of the model as a computational directed acyclic graph $G^{LM}$ in which information flows from the input tokens to the output logits through intermediate neuron activations. Following~\cite{jafari2025relp}, we define attention heads and MLP modules as nodes, with directed edges specifying how the output of one node is passed to another. 
As shown in Figure~\ref{fig:CURE}, the input of a node $v_1$ is defined as the sum of the outputs of all nodes with edges pointing to $v_1$, and each edge $e=v_1 \rightarrow v_2$ represents a direct computational dependency. A circuit is defined as a subgraph that connects the input tokens to the output logits.

Given a sample $(\inpori,y_u)\in\all$, our goal is to extract the influential circuits that are faithful to model prediction. Directly evaluating the importance of individual nodes is often insufficient, as it ignores how information is propagated and combined across modules. Instead, we assign a strength score to each edge $e=v_1 \rightarrow v_2$ that quantifies its contribution to driving the prompt $\inpori$ toward the target outcomes $y_u$. Specifically, we use the change in output probability $\Delta(\inpori)=\LLMori(Yes|\inpori)-\LLMori(No|\inpori)$ to measure variations in the model’s prediction. Following~\cite{syed2024attribution}, for an edge $e \in G^{LM}$, we evaluate its impact to metrics $\Delta$ by intervening on the activation input transmitted through this edge:
\begin{equation}
\label{eq:Causal Intervention}
\begin{aligned}
I(e)=|\Delta(x_u|\mbox{do}(e))-\Delta(x_u)|
\end{aligned}
\end{equation}
We adopt the do-notation from causal inference to emphasize that this intervention modifies the information flow along a specific edge.
Although existing methods such as ACDC~\cite{conmy2023towards,hanna2024have}can be used to evaluate causal impact, they are computationally inefficient. Moreover, user–item relationships naturally form a graph structure, which can be further exploited to improve both efficiency and localization precision. Based on this observation, we propose two methods.

\subsubsection{Activation Intervention}
To quantify $I(e)$ in Eq.~\ref{eq:Causal Intervention}, we directly intervene on the information flow along edge $e=v_1 \rightarrow v_2$. 
Concretely, we mask this edge by setting the message passed from $v_1$ to $v_2$, denoted as $m_{v_1 \to v_2}=0$, to zero, while keeping all other activations unchanged. However, modern LLMs involve an enormous number of edges to evaluate, making exact interventions time-consuming. Hence, we approximate the resulting change in the metric by linearly expanding $\Delta$ with respect to the input activation using a first-order Taylor expansion, which yields an efficient estimate of the marginal contribution of a single edge $e=v_1 \rightarrow v_2$ in the computational graph $G^{LM}$
\begin{equation}
\label{eq:ActivationIntervention}
\Delta(x_u \mid \mathrm{do}(m_{v_1\to v_2}=0))
\approx
\Delta(x_u)
-
m_{v_1\to v_2}^{\top}
\frac{\partial \Delta(x_u)}{\partial v_2} .
\end{equation}
Here, with a slight abuse of notation, we use $v_2$ to be the activation input of this node. The term $m_{v_1\to v_2}^{\top}
\frac{\partial \Delta(x_u)}{\partial v_2}$
captures the first-order effect of erasing the message of $e$, where the gradient is evaluated at the original (non-intervened) forward pass.
According to Eq.~\ref{eq:ActivationIntervention}, computing $I(e)$ for all edges only requires one forward pass to record edge messages and one backward pass to obtain their gradients.


\subsubsection{Activation Patching}
Although effective, the \textit{Activation Intervention} tends to over-select modules as the input $\inpori$ grows longer, since it captures the contributions of all tokens to $y_u$ in a global manner. To localize the responsible circuits, we therefore leverage activation patching~\cite{zhang2023towards, syed2024attribution} to measure module importance through contrasting activations that lead to different prediction outcomes. 
Following~\cite{syed2024attribution}, we construct a corrupt sample $\inpcf$ that shares the same task schema as $\inpori$ but elicits a distinct output from $\LLMori$. The importance of an edge is evaluated by replacing its activation $m_{v_1\to v_2}$ induced by $\inpori$ with that $m^{*}_{v_1\to v_2}$ from $\inpcf$ during the forward pass, while keeping all other activations unchanged. A significant change in the output metric $\Delta$ indicates that the edge plays a critical role in driving the prediction.
Formally, to quantify Eq.~\ref{eq:Causal Intervention}, we linearly approximate the importance score $I(e)$ by expanding $\Delta$ as a Taylor series with respect to the edge activation:
\begin{equation}
\label{eq:ActivationPatching}
\Delta(\inpori \mid \mathrm{do}(m_{v_1\to v_2}=m^{*}_{v_1\to v_2}))
\approx
\Delta(x_u)
+
(m^{*}_{v_1\to v_2}-m_{v_1\to v_2})^{\top}
\frac{\partial \Delta(x_u)}{\partial v_2} .
\end{equation}
The second term on the right-hand side serves as an estimate of $I(e)$.
However, unlike circuit detection in reasoning tasks, constructing the corrupt sample $\inpcf$ in LLMRec is non-trivial. Two challenges arise: 1) The input $\inpori$ can be substantially longer, as it summarizes a user’s historical interactions. 2) Prior work~\cite{hanna2024have,syed2024attribution} requires $\inpori$ and its corrupt counterpart $\inpcf$ to differ minimally at the input level while inducing a significant change in prediction. As a result, it is difficult to pinpoint which tokens in $\inpori$ should be modified to construct $\inpcf$.  Formally, $\inpcf$ should satisfy the following constraints
\begin{equation}
\label{eq:counterfactual}
\inpcf
=
\arg\min_{\mathbf{x}}\ \Delta(\mathbf{x})
\quad \text{s.t.}\quad
\|\mathbf{x}-\inpori\|_{0}\le K .
\end{equation}
Here, $K$ controls the number of items to be replaced. In our setting, we set $K=1$, meaning that only a single item is allowed to be replaced in $\inpori$. In addition, we require the inserted item to remain within the user’s interest distribution to avoid introducing out-of-distribution corrupt sample $\inpcf$.

The original input $\inpori$ consists of the user history $\mathcal{H}_u$ and the target item $i_t$, whose relationships can be captured by the user–item graph $\mathcal{G}$. Fortunately, this representation allows us to exploit rich structural information to construct appropriate corrupt samples $\inpcf$.
To efficiently generate $\inpcf$, we first identify the history item in $\mathcal{H}_u$ that is most influential to predicting $i_t$ within the context of $\inpori$, and then replace it with a weakly related item.
\paragraph{\textbf{Step1: Graph-based Item Scoring}}
We first leverage \textit{Personalized PageRank} (PPR)~\cite{yang2024efficient,li2023everything} on $\mathcal{G}$ to estimate item proximity with respect to a given user $u$.
Let $\boldsymbol{\pi}_u \in \mathbb{R}^{|\mathcal{V}|}$ denote the PPR vector associated with user $u$, where each entry $\boldsymbol{\pi}_u[i]$ measures the graph-based proximity of item $i$ to $u$.
The PPR vector is defined as the stationary solution of the following equation:
\begin{equation}
\label{eq:ppr}
\boldsymbol{\pi}_u
=
\alpha\, P^\top \boldsymbol{\pi}_u
+
(1-\alpha)\boldsymbol{p}_u,
\end{equation}
where $P$ denotes the transition matrix of the user--item graph and $\alpha \in (0,1)$ is a decay factor.
The \emph{preference vector} $\boldsymbol{p}_u$ is a probability distribution that encodes the user’s preference over items. Apart from the algebraic definition, PPR also has an intuitive random-walk interpretation: starting from user $u$, the walk proceeds as follows: at each step, it either (i) moves to a neighboring node on $\mathcal{G}$ according to the transition matrix $P$ with probability $\alpha$, or (ii) jumps to a node sampled from the preference distribution $\boldsymbol{p}_u$. Then $\boldsymbol{\pi}_u$ is defined as the stationary distribution over $\mathcal{V}$ after infinitely many steps.
Given the input $\inpori$, we further bias $\boldsymbol{\pi}_u$ toward items $i\in \mathcal{H}_u$ that contribute most to the model prediction.
Accordingly, we initialize the preference vector $\boldsymbol{p}_u$ as:
\begin{equation}
\boldsymbol{p}_u(i)=
\begin{cases}
\displaystyle
\frac{\exp(\tau S(i))}{\sum_{j\in \mathcal{H}_u}\exp(\tau S(j))},
& v\in \mathcal{H}_u,\\[6pt]
0,
& \text{otherwise},
\end{cases}
\end{equation}
where $S(i)$ measures the importance of item $i$ to the prediction and $\tau$ is a temperature parameter.
 Since an item $v$ may correspond to multiple tokens in the input, we approximate its importance by aggregating token-level gradients:
\begin{equation}
S(i)=\sum_{t\in \text{tokens}(i)}
\left\lVert
\frac{\partial \Delta(x_u)}{\partial \mathbf{e}_t}
\right\rVert,
\end{equation}
where $\mathbf{e}_t$ denotes the embedding of token $t$.

In practice, we precompute approximate PPR vectors for individual items offline~\cite{yang2024efficient,zhang2024towards}.
Specifically, we run an approximate PPR algorithm with a one-hot preference vector for each item $i$ and obtain the corresponding vector $\boldsymbol{\pi}^{pre}_i$.
Given an input $\inpori$, the personalized PPR vector can be efficiently constructed as a weighted sum:
\begin{equation}
\boldsymbol{\pi}_u
=
\sum_{i\in \mathcal{H}_u}
\frac{\exp(\tau S(i))}{\sum_{j\in \mathcal{H}_u}\exp(\tau S(j))}
\boldsymbol{\pi}^{pre}_i,
\end{equation}
which follows from the linearity of Personalized PageRank~\cite{yang2024efficient} with respect to the preference vector.
In this way, personalized item scores can be efficiently obtained without performing any online graph diffusion.
\paragraph{\textbf{Step 2: Replacing Items in $\inpori$}}
As LLMRec may not perfectly align with the underlying graph structure, the solution induced by PPR can be suboptimal for $\LLMori$.
Moreover, exhaustively searching over all possible replacements is computationally prohibitive.
We therefore restrict the search space to a small set of candidate substitutions when constructing the corrupt input $\inpcf$.

Given an input $\inpori$, we first select the top-3 most influential items from the user history $\mathcal{H}_u$ according to the gradient-based importance score $S(v)$, and construct a preferred item set $\mathcal{I}_u$ by selecting the top-50 items ranked by the PPR vector $\boldsymbol{\pi}_u$.
Following~\cite{yang2024efficient}, differences between entries in $\boldsymbol{\pi}_u$ reflect relative proximity under the same source $u$; thus, for each influential item, we choose the 10 least relevant items from $\mathcal{I}_u$ as replacement candidates. Overall, this procedure yields at most 30 candidate corrupt inputs, which are evaluated by the LLM to select the final counterfactual input $\inpcf$ according to Eq.~\ref{eq:counterfactual}.


\subsubsection{Greedy Circuit Discovery}
\label{sec:greedy}
After scoring edges via Equation~\ref{eq:Causal Intervention}, we employ the greedy extraction strategy of~\cite{syed2024attribution}. Starting from the logits, we iteratively add the highest-scoring edge whose child node is already in the circuit, constructing a top-down circuit while avoiding childless nodes. This process constructs a complete circuit in a top-down manner. The resulting procedure resembles a maximization variant of Dijkstra’s algorithm, with circuit complexity controlled by the number of iterations.

\subsection{Selective Circuits Updating}
\label{sec:selupdate}
For each sample in $\forget$, we select nearby samples from the remaining data to construct a retain buffer of size $|\remain| = k|\forget|$ (typically $k=6$). For each user–item interaction to be removed, we measure structural proximity on the user–item graph $\mathcal{G}$ using Personalized PageRank (PPR) with forget nodes as sources. Retain-set edges are ranked by their PPR scores, where higher scores indicate stronger proximity to the forget set. We select the top-$k$ nearest training samples to form the retain set, enabling unlearning by refining the decision boundary between closely related forget and retain data.

Here we borrow ideas from SCRUB~\cite{kurmanji2023towards} to construct the unlearning loss, and formulate it by enforcing two essential properties. Specifically, for each deleted sample $(\inpori, y_u) \in \all$:
\begin{itemize}
    \item \textit{Deviating the prediction from the original model $\LLMori$ on $\forget$}, where deleted samples are encouraged to produce predictions under $\un$ that explicitly differ from those of the original model, ensuring that the removed information is no longer preserved.\\ 
$\mathcal{L}_{F} = -D_{KL}(\LLMori(Yes|\inpori)||\un(Yes|\inpori))$. 
    \item \textit{Maintaining the prediction of $\LLMori$ on $\remain$}, where the selected retain samples $(\mathbf{x}_{r}, y_r) \in \remain$ associated with each deleted sample $(\inpori, y_u)$ are encouraged to produce similar predictions under $\un$ as those from the original model $\LLMori$, thereby preserving the model’s utility on retained data.\\ 
    $\mathcal{L}_{R} = \frac{1}{|\remain|}D_{KL}(\LLMori(Yes|\mathbf{x}_{r})||\un(Yes|\mathbf{x}_{r}))$.
    
 \end{itemize}   
Here, $D_{\mathrm{KL}}$ denotes the KL-divergence between the two distributions. The overall unlearning objective is defined as
$\mathcal{L}=\wretain\mathcal{L}_{R}+\wforget\mathcal{L}_{F}$; During optimization, all parameters of the original model $\LLMori$ are frozen, and only the parameters of the unlearned model $\un$, which is initialized from $\LLMori$, are updated.

For each sample in $\forget$ and $\remain$, we apply the method in Section~\ref{sec:Circuits} to extract the retain and forget circuits, denoted as $\mathcal{C}_f$ and $\mathcal{C}_r$, respectively. 
Based on the detected circuits, we categorize the associated modules into different functional groups according to their roles along the circuits. Neurons within each modules are assigned to the same group and updated using group-specific policies, allowing each group to be optimized for its intented function without interference. 

The update policies are defined as follows:
\begin{itemize}
  \item \textit{Forget-Specific Neurons} $\Theta_f$:
A neuron $\theta_f \in \Theta_f$ is identified as forget-specific if 
$\theta_f \in \mathcal{C}_f \cap \mathcal{C}_r^{c}$, where $\mathcal{C}_r^{c}$ denotes the the complement of $\mathcal{C}_r$.
indicating that it contributes primarily to samples in $\forget$ rather than $\remain$. 
Accordingly, these neurons are updated solely with respect to the forget loss to enhance data removal. 
We define the gradient as $\gradfforget=\nabla_{\theta_f} \mathcal{L}_{F}$ and update
\[
\theta_f \leftarrow \theta_f - \alpha \, \gradfforget.
\]
 \item \textit{Retain-specific Neurons} $\Theta_r$:
Similarly, a neuron $\theta_r \in \Theta_r$ is identified as retain-specific if 
$\theta_r \in \mathcal{C}_r \cap \mathcal{C}_f^{c}$. As these neurons are essential for preserving the model utility, we update them only with respect to the retain loss. We define the gradient as $\gradrretain=\nabla_{\theta_r} \mathcal{L}_{R}$ and update:
\[
\theta_r \leftarrow \theta_r - \alpha \,\gradrretain .
\]

   \item \textit{Function-shared Neurons} $\Theta_{sh}$. A neuron $\theta_{sh} \in \Theta_{sh}$ is identified as function-shared if 
$\theta_{sh} \in \mathcal{C}_r \cap \mathcal{C}_f$, indicating that it is involved in high-level functions contributing to both data removal and utility maintenance. 
To alleviate gradient conflicts during optimization, CURE adopts a simple projection-based strategy~\cite{yu2020gradient,shi2023recon}. 
Specifically, when the gradients of the two objectives conflict, i.e., 
$(\gradrshare)^{\top}\gradfshare < 0$, 
we iteratively select one objective from $\mathcal{L}_r$ and $\mathcal{L}_f$ and project its gradient onto the other, removing destructive components. When $\mathcal{L}_r$ is selected, the update is defined as:
\[
\begin{aligned}
\mathbf{g}_{sh} 
&= \omega_{R}(\gradrshare- \frac{(\gradrshare)^{\top} \gradfshare}
{\|\gradfshare\|^{2}} \, \gradfshare)+ \omega_{F}(\gradfshare-\frac{(\gradfshare)^{\top} \gradrshare}
{\|\gradrshare\|^{2}}), \\
\theta_{sh} 
&\leftarrow \theta_{sh} - \alpha \mathbf{g}_{sh}.
\end{aligned}
\]

 \end{itemize}

\section{Theoretical Analysis}
\label{sec:Theory}
We provide a theoretical analysis of CURE to indicate that a single gradient update on the model parameters of CURE achieves lower loss than normal gradient descent. Let $\Theta=\{\orishare, \oriforget, \oriretain\}$ denote the parameters of components in selected circuits. After applying CURE, the model parameters are $\hat{\Theta}=\{\selshare,\selforget, \selretain \}$. An one step gradient update of $\hat{\Theta}$ is:
\begin{equation}
\label{eq:CUREupdate}
\begin{aligned}
\hat{\theta}_{sh}
&= \theta_{sh}
- \alpha\Big[
\wretain \Big(
\gradrshare
- \frac{\gradrshare \cdot \gradfshare}{\|\gradfshare\|^2}\,\gradfshare
\Big)
+ \wforget \Big(
\gradfshare
- \frac{\gradfshare \cdot \gradrshare}{\|\gradrshare\|^2}\,\gradrshare
\Big)
\Big], \\
\hat{\theta}_f &= \theta_f - \alpha\gradfforget,
\qquad
\hat{\theta}_r = \theta_r - \alpha\gradrretain .
\end{aligned}
\end{equation}
Without applying CURE, the parameters are
$\Tilde{\Theta}=\{\nshare\,\nforget, \nretain\}$. An one-step gradient update is given by:
\begin{equation}
\label{eq:SGDupdate}
\begin{aligned}
\tilde{\theta}_i &= \theta_i - \alpha(\wretain \mathbf{g}_{i}^{R} + \wforget \mathbf{g}_{i}^{F}), \quad i \in \{sh, f, r\}
\end{aligned}
\end{equation}
Then, we have the following theorem. 
\begin{theorem}
Assume that the joint loss function is defined as
$L(\Theta) = \wretain L_R(\Theta) + \wforget L_F(\Theta)$ and $\wretain+\wforget=1$. Then for any sufficiently small learning rate $\alpha>0$, we have
\begin{equation}
\label{eq:Objective}
\begin{aligned}
L(\hat{\Theta})\leq L(\Tilde{\Theta})
\end{aligned}
\end{equation}
\end{theorem}

\noindent\textbf{Proof.} We iteratively update each parameter group while fixing the others. The loss difference between the conventional update $\tilde{\params}$ and CURE $\hat{\params}$ is analyzed via two components, $A$ and $B$:
\begin{equation}
\begin{aligned}
    &\Loss(\Tilde{\Theta}) - \Loss(\hat{\Theta}) \approx (\nshare - \selshare)^\top (\gradrshare + \gradfshare) \\
    &\quad + (\nforget - \selforget)^\top \gradfforget + (\nretain - \selretain)^\top \gradrretain \\
    &= \underbrace{-\Big[ \gradrshare \cdot \gradfshare + (\gradrshare \cdot \gradfshare)^2 \big( \tfrac{\wretain}{\|\gradfshare\|^2} + \tfrac{\wforget}{\|\gradrshare\|^2} \big) \Big]}_{A} \\
    &\quad + \underbrace{\wretain(\gradfforget - \gradrforget)^\top \gradfforget + \wforget(\gradrretain - \gradfretain)^\top \gradrretain}_{B}
\end{aligned}
\end{equation}

\noindent\textbf{For Part A:}
Let $\gamma = \|\gradrshare\| / \|\gradfshare\|$ and $\psi$ be the angle between $\gradrshare$ and $\gradfshare$.
$A \ge 0$ is guaranteed if and only if the condition holds:
\begin{equation}
\label{equ:Acondition}
    \cos \psi \ge -\frac{\gamma}{\wretain + (1 - \wretain)\gamma^2}
\end{equation}
When $\gamma$ is close to $1$, indicating that the two gradients have comparable magnitudes, equation~\ref{equ:Acondition} holds with high probability.
Notably, when $\gamma = 1$, the condition reduces to $\cos \psi \ge -1$, which is \emph{always satisfied}.
In practice, we control $\gamma$ through gradient normalization, ensuring that the magnitudes of the two gradients remain comparable.

\noindent\textbf{For Part B:} We consider neurons in each task-specific component, and denote $\Delta m=(m^{*}_{v_1\to v_2}-m_{v_1\to v_2})$ in equation~\ref{eq:ActivationPatching}.
According to the loss functions $\mathcal{L}_{F}$ and $\mathcal{L}_{R}$ defined in Section~\ref{sec:selupdate}, we observe that both losses are monotonic with respect to
$\Delta(\inpori)=\LLMori(Yes|\inpori)-\LLMori(No|\inpori)$, which is used to detected the forget and retain circuits. which is used to detect forget and retain circuits. Hence, the greedy selection policy in Section~\ref{sec:greedy} satisfies:
\begin{equation}
\begin{aligned}
    & \|\Delta m \gradfforget\|^2 \ge \|\Delta m \gradrforget\|^2, \quad \|\Delta m \gradrretain\|^2 \ge \|\Delta m \gradfretain\|^2 \\
    & \implies 2\|\Delta m \gradfforget\|^2 \ge \|\Delta m \gradfforget\|^2 + \|\Delta m \gradrforget\|^2 \\
    & \quad \ge \|\Delta m\|^2 (\gradfforget \cdot \gradrforget) \implies \|\gradfforget\|^2 - \gradfforget \cdot \gradrforget \ge 0
\end{aligned}
\end{equation}
This leads to $\|\gradfforget\|^2 - \gradfforget \cdot \gradrforget > 0$, ensuring the benefit of targeted forgetting. An analogous result can be derived $\|\gradrretain\|^2 - \gradfretain \cdot \gradrretain > 0$
\section{Experiments}

\begin{figure*}[t]
    \centering
    \includegraphics[width=1.0\textwidth]{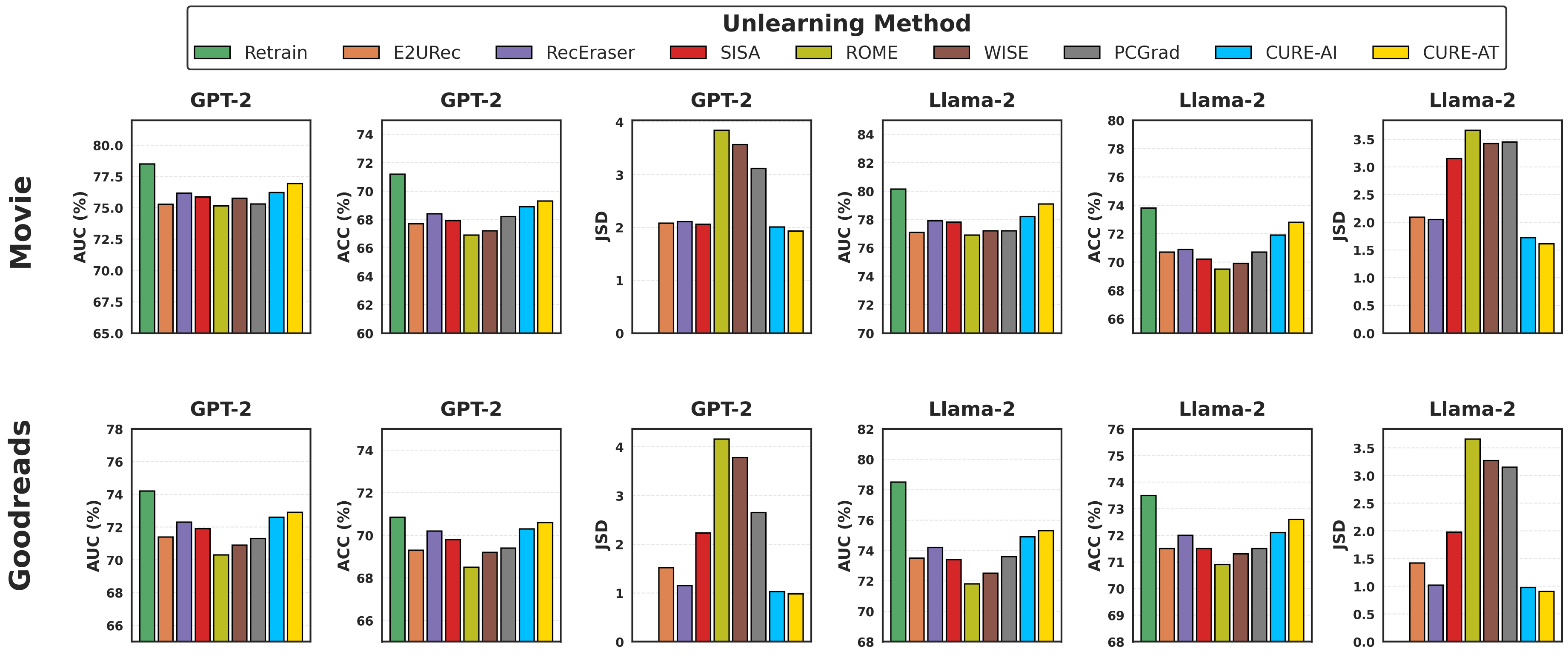}
    \caption{Unlearning effectivenss and model performance on Movie and Goodreads. }
    \label{fig:MainResult}
\end{figure*}

To evaluate the effectiveness of our proposed method, we conduct a series of experiments to address the following research questions:
\begin{itemize}
\item \textbf{RQ1:} How effective is CURE in achieving unlearning while preserving recommendation utility?
\item \textbf{RQ2:} How well does CURE mitigate gradient conflicts during unlearning?
\item \textbf{RQ3:} How do different components of CURE contribute to its overall effectiveness?
\end{itemize}
\subsection{Experimental Settings}
\noindent{\textbf{Datasets}}
We evaluate our proposed method on two widely recognized recommendation benchmarks: \textbf{MovieLens-1M (ML-1M)} and \textbf{GoodReads (GD)}. \textbf{MovieLens-1M (ML-1M)} contains movie metadata and user ratings. Following~\cite{wang2025towards}, we transform the original ratings into binary labels for LLM prompting, where ratings greater than 3 are treated as positive (mapped to “Yes”), and the remaining ratings are mapped to “No”. Similarly, \textbf{GoodReads (GD)} comprises book features and user ratings, and we apply the same binarization strategy.

\subsubsection{Implementation Details}
All baseline models and our proposed framework, \textbf{CURE}, utilize the same backbone architectures: \textbf{GPT-2} \cite{brown2020language} and \textbf{Llama-2 (7B)} \cite{touvron2023llama}. We follow standard hyperparameter configurations established in recent LLM-based recommendation literature. 

We evaluate two variants, \textbf{CURE-AI} and \textbf{CURE-AP}, which adopt Activation Intervention and Activation Patching for circuit extraction, respectively. CURE-AI identifies coarse-grained circuits, while CURE-AP achieves higher-precision localization by constructing corrupt samples based on the original input.
 
For the unlearning experiments, we follow the protocol in~\cite{wang2025towards}. Both datasets are split into training, validation, and test sets with a ratio of 7:2:1 to obtain the original model, and $20\%$ of the training data are designated as the forget set during unlearning. Regarding the specific hyperparameters for CURE, we set the intervention threshold based on the top-5\% of activation weights. For the second stage, we utilize the AdamW optimizer with a learning rate of $5 \times 10^{-5}$. To facilitate efficient training on Llama-2 (7B), we employ Low-Rank Adaptation (LoRA) with rank $r=8$. All experiments are conducted on NVIDIA A100 GPUs, and hyperparameters are tuned on the validation set to ensure optimal performance.


\subsubsection{Baselines}
\smallskip We compare \textsf{CURE} against several baselines, including exact unlearning paradigms and model editing techniques: $(i)$ \textbf{Retrain}, which refers to training from scratch without target forgotten data, 
$(ii)$ \textbf{SISA}~\cite{bourtoule2021machine}, a partition-based retraining method that aggregates predictions from sub-models trained on separate data shards;
$(iii)$ \textbf{RecEraser}~\cite{li2023ultrare}, a recommendation-specific unlearning method that preserves collaborative information through specialized partitioning;
$(iv)$ \textbf{E2URec}~\cite{wang2025towards}, an efficient unlearning framework that utilizes a teacher-student architecture to guide the unlearning process;
$(v)$ \textbf{ROME}~\cite{meng2022locating}, a model editing method that performs direct parameter updates at early layers;
$(vi)$ \textbf{WISE}~\cite{wang2024wise}, a baseline that introduces learnable parameters at later layers to incorporate new information or forget data; $(vii)$ \textbf{PCGrad}~\cite{yu2020gradient}, a gradient surgery approach for gradient collicts that projects gradients onto the normal plane of conflicting tasks.
\subsubsection{Evaluation Settings}
Our goal is to achieve precise and efficient unlearning for LLMRec while preserving recommendation utility. Following standard protocols in recommendation unlearning, we evaluate \textsf{CURE} based on the following three dimensions:

\begin{itemize}
\item \textbf{Recommendation Performance:} To evaluate the recommendation performance, we use \textbf{Area Under the ROC Curve (AUC $\uparrow$)}, \textbf{Accuracy (ACC $\uparrow$)}, and \textbf{LogLoss (LL $\downarrow$)} on test set.
\item \textbf{Unlearning Effectiveness:} To quantify how closely the unlearned model aligns with the gold-standard \textit{Retrain} model, we compute the \textbf{Jensen-Shannon Divergence (JSD $\downarrow$)} between their respective output probability distributions on the forgotten data. Lower JSD values signify a more effective elimination of the target data, indicating that the unlearned model's behavior successfully mimics a model that never encountered the forgotten samples.
\item \textbf{Unlearning Efficiency:} We measure the unlearning efficiency by \textbf{Unlearning Time ($\downarrow$)} measuring the total wall-clock time required for the unlearning process. 
\end{itemize}
\label{sec:experiments}


\subsection{Main Results (RQ1)}

We evaluate unlearning performance on two recommendation benchmarks, MovieLens-1M and GoodReads, using GPT-2 and LLaMA2-7B as backbone models. Following previous section, we assess (i) recommendation utility via AUC and ACC, (ii) unlearning effectiveness via JS-Divergence (JSD) between the unlearned and retrained models on the forgotten data, and (iii) efficiency via unlearning time.

Across all settings, results consistently show that:
(i) full retraining remains the upper bound in recommendation accuracy but is computationally prohibitive;
(ii) \textbf{CURE-AT} achieves the best overall trade-off, preserving near-retraining utility while attaining the lowest divergence from retrained models with orders-of-magnitude lower cost;
(iii) efficiency-oriented methods such as ROME and WISE sacrifice unlearning completeness, as reflected by substantially higher JSD; and
(iv) these trends remain stable across model scales and datasets, indicating strong robustness of the proposed approach.

\noindent \textbf{\textit{Overall Results}}
Across both backbones and datasets, we observe a clear trade-off between utility preservation and unlearning effectiveness. While approximate unlearning methods reduce computational cost, many incur either significant performance degradation or incomplete forgetting. In contrast, CURE-AT and CURE-AI consistently strike a favorable balance, achieving strong recommendation performance with minimal divergence from retraining.

Figure~\ref{fig:MainResult} reports results on MovieLens-1M using GPT-2 Large as the backbone. Most baseline unlearning methods (E2URec, RecEraser, SISA, ROME, WISE,PCGrad) exhibit noticeable drops in AUC and ACC relative to retraining, with JSD values generally exceeding 2.0. In contrast, CURE-AT achieves the lowest divergence (JSD = 1.93) while preserving higher recommendation quality (AUC = 76.93, ACC = 69.3). CURE-AI follows closely, outperforming all other baselines in both unlearning effectiveness and efficiency.


Figure~\ref{fig:MainResult} summarizes results using LLaMA2-7B. Compared to GPT-2 Large, all methods benefit from increased model capacity, yielding higher absolute AUC and ACC. However, the relative ranking of unlearning methods remains consistent. CURE-AT again achieves the strongest performance among non-retraining approaches (AUC = 79.1, ACC = 72.8) while also attaining the lowest divergence (JSD = 1.61). Although RecEraser and SISA partially preserve utility, their substantially longer unlearning times make them impractical for frequent or large-scale unlearning scenarios.

Figure~\ref{fig:MainResult} presents results on the GoodReads dataset using GPT-2 Large. Compared to MovieLens-1M, GoodReads exhibits a more challenging unlearning setting, with larger variance in JSD across methods. While retraining again yields the highest accuracy (AUC = 74.2, ACC = 70.85), CURE-AT achieves near-retraining performance (AUC = 72.9, ACC = 70.6) with the lowest divergence (JSD = 0.98). CURE-AI follows closely (JSD = 1.03), outperforming RecEraser and SISA in both effectiveness and efficiency.

Baseline methods such as ROME and WISE show particularly high divergence (JSD $>$ 3.7), indicating severe forgetting failure despite their fast execution. PCGrad improves over several baselines but remains inferior to CURE-based approaches in both utility preservation and unlearning completeness.

Figure~\ref{fig:MainResult} reports corresponding results using LLaMA2-7B. Similar trends persist at larger model scale: CURE-AT achieves the best trade-off between accuracy (AUC = 75.3, ACC = 72.6) and unlearning effectiveness (JSD = 0.91), substantially narrowing the gap to retraining while maintaining low computational cost. Other methods either exhibit higher divergence or incur significantly longer unlearning times, reinforcing the robustness of CURE-AT across datasets and model sizes.

\noindent \textbf{\textit{Unlearning Time Comparison}} Table~\ref{tab:time_avg_aligned} compares the efficiency of different unlearning methods. Overall, methods exhibit large variance in computational efficiency. Retraining is the most time-consuming method, incurring prohibitively expensive cost than other unlearning methods. RecEraser and SISA are consistently the most expensive approaches, requiring up to tens of thousands of seconds on LLaMA2-7B. While these methods partially preserve recommendation utility, their reliance on repeated retraining or multiple model shards makes them impractical for large-scale or frequent unlearning scenarios. In contrast, \textbf{CURE-AI} and \textbf{CURE-AT} consistently achieve strong efficiency across both MovieLens-1M and GoodReads, with unlearning times comparable to or faster than ROME and WISE, while maintaining significantly lower JSD divergence and higher recommendation accuracy. Notably, \textbf{CURE-AI} is $\mathbf{3.5}\times$ faster and \textbf{CURE-AT} is $\mathbf{3.3}\times$ than E2URec. Although slightly slower than WISE and ROME due to their localized updating strategy, CURE achieves significantly better performance.


\begin{table}[t]
\centering
\caption{Unlearning time comparison (s) $\downarrow$.}
\label{tab:time_avg_aligned}
\begin{tabular}{l|cc|cc}
\toprule
& \multicolumn{2}{c|}{\textbf{GoodReads}} & \multicolumn{2}{c}{\textbf{MovieLens-1M}} \\
Method & GPT-2 & LLaMA-2 & GPT-2 & LLaMA-2 \\
\midrule
Retrain        & 15,800  & 128,000 & 29,300  & 132,300 \\
\midrule
E2URec         & 1,900   & 13,600  & 3,200   & 21,800  \\
RecEraser      & 4,200   & 32,500  & 7,100   & 49,700  \\
SISA           & 3,700   & 29,800  & 6,800   & 45,300  \\
ROME           & 600     & 3,600   & 1,700   & 12,100  \\
WISE           & 800     & 7,200   & 1,200   & 9,800   \\
PCGrad         & 2,600   & 16,700  & 3,600   & 23,300  \\
\midrule
CURE-AI (Ours) & 600     & 3,900   & 1,500   & 10,200  \\
CURE-AT (Ours) & 600     & 4,100   & 1,700   & 11,300  \\
\bottomrule
\end{tabular}
\end{table}

\subsection{Effectiveness in Mitigating Gradient Conflicts(RQ2)}
\noindent{\textbf{CURE greatly reduces the occurence of conflicting gradients}} As shown in Figure~\ref{fig:frequency}, we compare the distribution of $\cos \psi$ before and after applying CURE on \textbf{MovieLens-1M (ML-1M)}. The results show that CURE substantially suppresses severely conflicting gradient pairs ($\cos \psi \in [-1, -0.02)$), reducing their proportion by at least $55\%$ and up to $76\%$ compared with E2URec. In contrast, PCGrad and ROME yields only marginal reductions, and in some cases even increase the proportion of conflicting gradients. An interesting observation is that WISE also exhibits fewer conflicting gradients, likely due to its restriction to updating later layers, which implicitly avoids interference with shared early representations.
\subsection{In-depth Analysis(RQ3)}

\subsubsection{Ablation Study}
We examine the effect of $\wretain$, which balances the retain and forget losses, as shown in Table~\ref{tab:ablation_omega}. The results shows the necessity of jointly optimizing $\mathcal{L}{\mathrm{F}}$ and $\mathcal{L}{\mathrm{R}}$ to achieve a favorable trade-off. When $\wretain \geq 0.6$, performance remains stable; however, for $\wretain \leq 0.4$, AUC drops sharply. Accordingly, we set $\wretain = 0.6$ in all experiments.

\subsubsection{Why does Circuit-aware Unlearning work} To validate the  of our circuit-aware approach, we compare CURE with two parameter-efficient unlearning baselines, ROME and WISE, by examining the modules they modify. As showen in Figure~\ref{fig:casestudy}, we observe that WISE primarily edits later layers, while ROME focuses on early layers. However, both methods operate on isolated modules within the LLM and fail to capture several critical components involved in the unlearning procedure. 
Instead, CURE identifies complete end-to-end circuits spanning all relevant modules, and propagates targeted modifications along the information flow from the input to the logits. Moreover, each module within an identified circuit attends to semantically relevant input tokens. 
For example, movies sharing the same genre as the target movie
(e.g., \emph{GoldenEye} and \emph{Congo}) are primarily attended to by \texttt{mlp8}, \texttt{mlp9}, and \texttt{mlp11}, enabling CURE to forget such interactions in a transparent and interpretable manner.
\begin{figure}
\centering
    \includegraphics[width=1.0\linewidth]{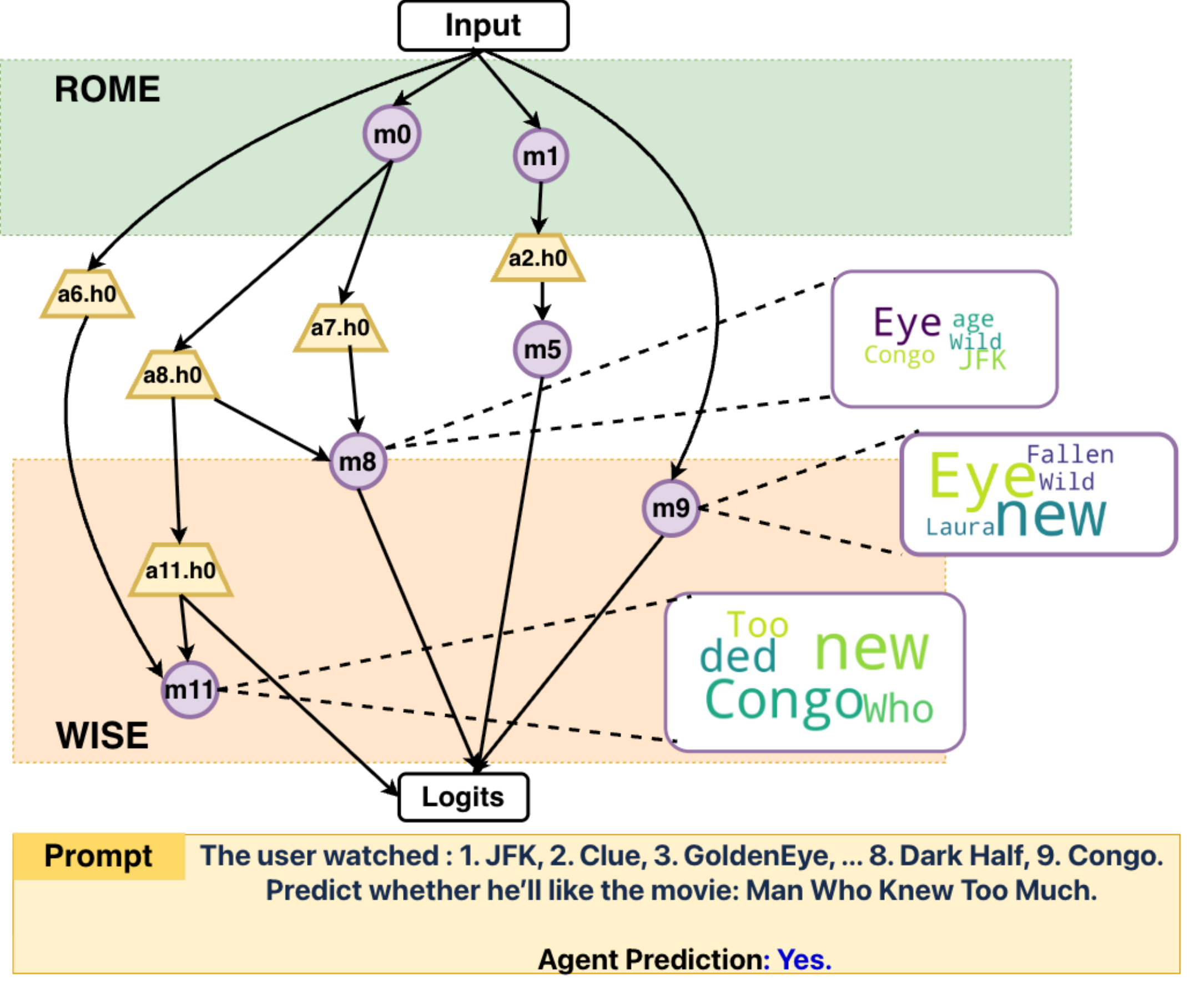}
    \caption{Transparent comparison of circuit-aware unlearning with ROME and WISE}

    \label{fig:casestudy}
\end{figure}

\begin{figure}
\centering
    \includegraphics[width=1\linewidth]{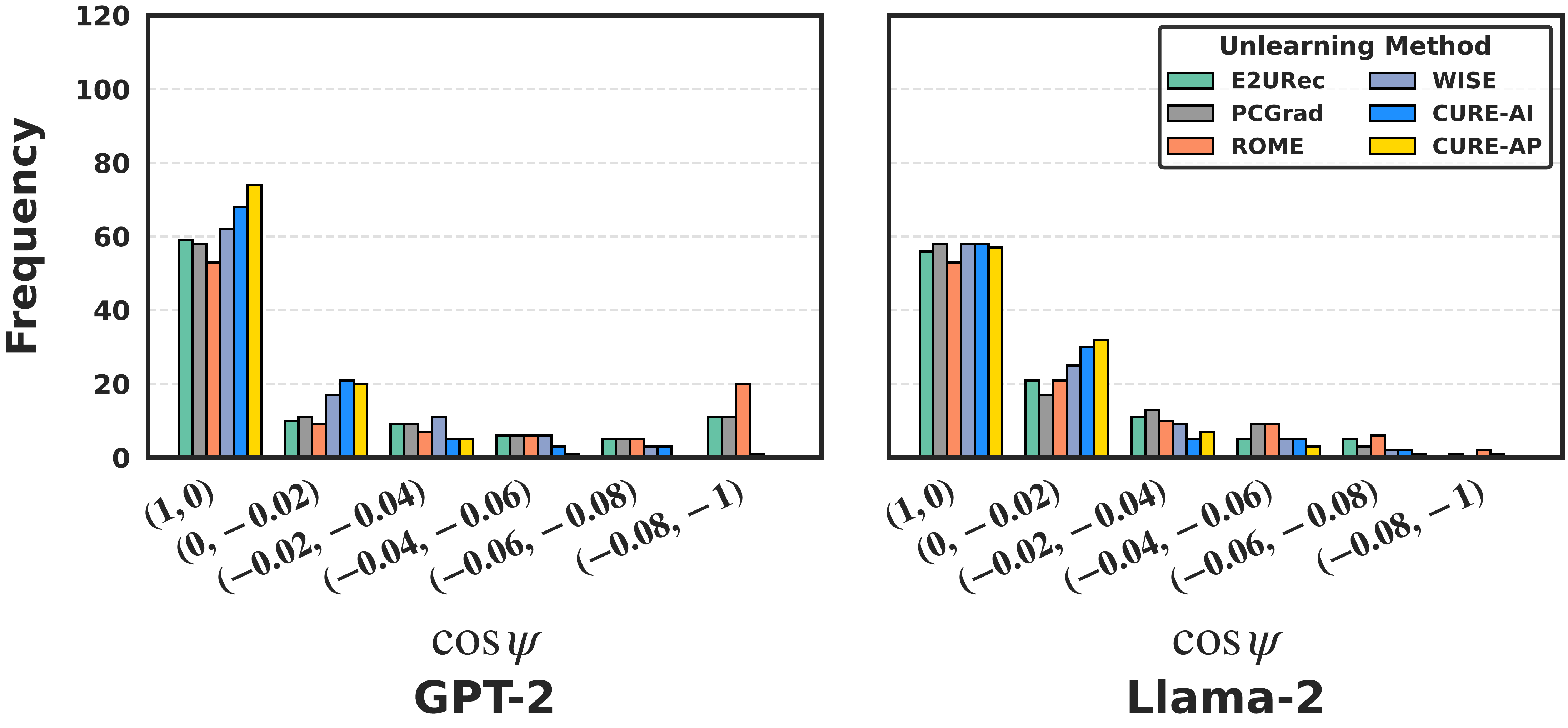}
    \caption{The distribution of gradient conflicts ($\cos \psi$) on GoodReads. The left and right columns use GPT-2 Large and Lamma-2 7B as backone, respectively }

    \label{fig:frequency}
\end{figure}

\begin{table}[t]
\small
\centering
\caption{Ablation study of $\omega_r$ on GoodRead across different backbone. Best performance is \textbf{bold} and second best is \underline{underlined}.}
\begin{tabular}{l|ccc|ccc}
\toprule
\multirow{2}{*}{$\omega_r$} & \multicolumn{3}{c|}{\textbf{Llama-2}} & \multicolumn{3}{c}{\textbf{GPT-2}} \\
\cmidrule{2-7}
 & AUC $(\uparrow)$ & ACC $(\uparrow)$ & JSD $(\downarrow)$ & AUC $(\uparrow)$ & ACC $(\uparrow)$ & JSD $(\downarrow)$ \\
\midrule
0.2 & 71.3 & 69.1 & \underline{0.91} & 69.1 & 67.5 & \textbf{0.94} \\
0.4 & 73.5 & 71.2 & \textbf{0.90} & 71.2 & 69.5 & \underline{0.98} \\
0.6 & \textbf{75.3} & \textbf{72.1} & \underline{0.91} & \textbf{72.9} & \textbf{70.6} & \underline{0.98} \\
0.8 & \underline{74.9} & \underline{71.8} & 0.95 & \textbf{72.9} & \underline{70.5} & 1.05 \\
\bottomrule
\end{tabular}
\label{tab:ablation_omega}
\end{table}

\section{Related Work}
\label{sec:Relate}

\subsection{Unlearning for LLM-based Recommendation.} 
Traditional recommendation unlearning methods, such as RecEraser \cite{chen2022recommendation} and AltEraser \cite{liu2022forgetting}, focus on partitioning collaborative data to safeguard privacy but are ill-suited for the massive parameter space of Large Language Models (LLMs). Conversely, general LLM unlearning often relies on approximate methods like gradient ascent \cite{yao2024large} or in-context label flipping, which frequently trigger catastrophic forgetting and degrade recommendation utility. To mitigate these issues, recent LLMRec-specific frameworks have adopted parameter-efficient fine-tuning (PEFT) \cite{zhang2025parameter}. E2URec\cite{wang2025towards} introduces a teacher-student architecture to guide the unlearning process via minimal LoRA updates, while the Adapter Partition and Aggregation (APA) framework employs data sharding and adapter retraining to achieve exact unlearning. However, these methods typically treat the model as a black box and apply uniform updates, leading to gradient conflicts between forgetting and retaining objectives. Our framework, \textsf{CURE}, addresses this by moving beyond uniform updates to a more granular, component-specific optimization strategy.

\subsection{Circuit Discovery}
Circuit discovery is the task of identifying sparse, functional subgraphs within a neural network that are causally responsible for implementing specific capabilities or behaviors \cite{conmy2023towards}. This paradigm shifts the focus from global parameter analysis to localizing the "essential computation" for a particular task. Early manual investigations, such as those by \cite{wang2022interpretability} and \cite{hanna2023does}, utilized causal mediation analysis and activation patching to uncover circuits for indirect object identification and mathematical reasoning in small-scale models. However, the manual search space grows exponentially with model depth, leading to the development of automated frameworks. ACDC \cite{conmy2023towards} automates circuit identification by recursively pruning edges based on their contribution to model faithfulness. While effective, its reliance on iterative testing renders it too computationally expensive for the scale of modern LLMs. To overcome this, Subnetwork Probing \cite{cao2021low} treats circuit identification as a mask-learning problem, optimizing for both fidelity and sparsity. More recently, \textbf{Edge Attribution Patching (EAP)} \cite{syed2024attribution} has emerged as a high-efficiency alternative, leveraging gradient-based importance scores to approximate the effect of interventions with minimal forward and backward passes. While these techniques have primarily been applied to linguistic benchmarks, we leverage circuit discovery techniques to extract the core circuits underlying item recommendation in LLM based recommendation. This structural decomposition provides the necessary transparency to disentangle the model into \textit{forget-specific} and \textit{retain-specific} modules~\citep{cheng2023gnndelete,fan2023salun,cheng2023multimodal}, bridging the gap between mechanistic interpretability and trustworthy unlearning in recommender systems.

\subsection{Gradient Conflicts in Unlearning}
The optimization of unlearning objectives often mirrors the challenges of multi-task learning, where a model must simultaneously satisfy competing goals. In the context of unlearning, a fundamental tension exists between the \textit{forgetting} objective (erasing specific data) and the \textit{retaining} objective (maintaining model utility) \cite{patel2025learning}. Previous studies have identified that these dual goals frequently lead to detrimental gradient interference, where the update direction for one task adversely impacts the other \cite{yu2020gradient}. 
While PCGrad \cite{yu2020gradient} introduced a model-agnostic ``gradient surgery'' approach—projecting conflicting gradients onto their respective normal planes—this method does not account for the unique functional structure of LLMs. Recent analysis in BLUR \cite{reisizadeh2025blur} shows that these gradient conflicts are even more severe in Large Language Models (LLMs). This complexity requires specialized optimization strategies to ensure that forgetting specific data does not accidentally damage the model's performance on remaining tasks. Our work, \textsf{CURE}, builds on these insights by using circuit discovery to physically isolate the parameters where these conflicts occur, allowing for a more targeted resolution than previous gradient-projection methods.

\section{Conclusion}
In this paper, we propose CURE, a circuit-aware unlearning framework that addresses the challenges of gradient conflicts in LLMRec unlearning. By leveraging mechanistic interpretability to disentangle computational circuits into functionally distinct modules, CURE enables precise, task-specific parameter updates that effectively remove sensitive information while preserving model utility. Our evaluation demonstrates that CURE outperforms state-of-the-art baselines with an $18\%$ improvement in unlearning efficiency and a $6\%$ gain in utility, while achieving a $3.5\times$ speedup. By shifting from black-box updates to transparent, circuit-level interventions, CURE provides a robust and efficient solution for privacy-preserving recommendation.

\bibliographystyle{ACM-Reference-Format}
\balance
\bibliography{references}










\end{document}